\documentclass{amsart}

\usepackage{amsmath,amsfonts,amssymb,amsthm,floatrow,subfig, graphicx,enumerate,xcolor,array,algcompatible,algorithm,afterpage}
\usepackage[para]{footmisc}
\usepackage[noend]{algpseudocode}
\usepackage[colorlinks,citecolor=blue,urlcolor=blue]{hyperref}

\usepackage{tikz}
\usetikzlibrary{positioning}
 \usepackage[round]{natbib}
\usepackage{float}
\floatstyle{plaintop}
\restylefloat{table}

\input xy
\xyoption{all}
\newtheorem{thm}{Theorem}[section]
\newtheorem{lem}[thm]{Lemma}

\theoremstyle{definition}

\theoremstyle{remark}

\makeatletter
\def\BState{\State\hskip-\ALG@thistlm}
\makeatother

\newcommand{\QED}{\ifhmode\unskip\nobreak\fi\quad {\rm Q.E.D.}} 

\newcommand{\R}{\mathbb{R}}

\newcommand{\sign}{\text{sign}}

\newcommand\numberthis{\addtocounter{equation}{1}\tag{\theequation}}
\newcommand{\CComment}[1]{{\color{commentcol}\Comment{\parbox[t]{.5\linewidth}{#1}}}}
\algnewcommand{\LineComment}[1]{\State {\color{commentcol}\(\triangleright\) \parbox[t]{.9\linewidth}{#1}}}

\numberwithin{equation}{section}

\definecolor{commentcol}{gray}{0.5}

\title[Efficient Computation of the Bergsma-Dassios Sign Covariance]{Efficient Computation of the Bergsma-Dassios Sign Covariance}

\author{Luca Weihs} 
\address{Department of Statistics, University
  of Washington, Seattle, WA, U.S.A.}
\email{lucaw@uw.edu}

\author{Mathias Drton} 
\address{Department of Statistics, University
  of Washington, Seattle, WA, U.S.A.}
\email{md5@uw.edu}

\author{Dennis Leung} 
\address{Department of Statistics, University
  of Washington, Seattle, WA, U.S.A.}
\email{dmhleung@uw.edu}

\date{\today}                                

\begin{document}

\begin{abstract}
	In an extension of {K}endall's $\tau$,  Bergsma and Dassios (2014) introduced a covariance measure $\tau^*$ for two ordinal random variables that vanishes if and only if the two variables are independent.  For a sample of size $n$, a direct computation of $t^*$, the empirical version of $\tau^*$, requires $O(n^4)$ operations.  We derive an algorithm that computes the statistic using only $O(n^2\log(n))$ operations.
\end{abstract}


\maketitle

\section{Introduction}

Kendall's $\tau$ \citep{kendall} and Spearman's $\rho$ \citep{spearman} are popular measures of dependence between two random variables $X$ and $Y$.  However, both have the undesirable property that they may be equal to zero even when $X$ and $Y$ are not independent. Addressing this weakness, \citet{bergsma2014} have defined a new coefficient, $\tau^*$, which, under mild conditions on the joint distribution of $(X,Y)$, is zero if and only if $X$ and $Y$ are independent.  However, a computational price is to be paid for this property as a na\"{\i}ve computation of $t^*$, the empirical version of $\tau^*$, requires $O(n^4)$ time for a sample of size $n$. 

In this paper we present an algorithm which computes $t^*$ in
$O(n^2\log(n))$ time, inspired by a similar improvement for computing
(the empirical version of) Kendall's $\tau$.   Indeed, by leveraging
binary tree algorithms and observing that Kendall's statistic depends
only on the relative order of data points, \citet{fast_kendall} showed
that Kendall's  $\tau$ can be computed in $O(n\log(n))$ time rather than $O(n^2)$. We follow a similar strategy by exploiting the fact that computing $t^*$ relies only on the relative ordering of quadruples of points. Due to excessive time requirements, Bergsma and Dassios limit their computational examples to sample sizes with $n\leq 50$ and suggest approximating $t^*$ by random subsampling for larger samples.  As will be shown in Section \ref{sec:sim}, our algorithm computes $t^*$ exactly in less than a second for sample sizes in the thousands.

\subsection{Background and Setup}

Given a sample $(x_1,y_1), ..., (x_n,y_n)$ of points in $\mathbb{R}^2$, define the statistic 
\begin{align}
\label{eq:t_star}
	t^* &:= \frac{(n-4)!}{n!}\sum_{\substack{1\leq i,j,k,l \leq n \\ i,j,k,l\ \text{distinct}}} a(x_i,x_j,x_k,x_l)a(y_i,y_j,y_k,y_l) 
\end{align}
where
\begin{align*}
	a(z_1,z_2,z_3,z_4) := \sign(|z_1-z_2| + |z_3-z_4| - |z_1-z_3| - |z_2-z_4|).
\end{align*}
Here $t^*$ is the U-statistic, U standing for unbiased, corresponding to the population coefficient $\tau^* := Ea(X_1,X_2,X_3,X_4)a(Y_1,Y_2,Y_3,Y_4)$ of \cite{bergsma2014} where $(X_1,Y_1),...,(X_4,Y_4)$ are random vectors drawn independently from some bivariate distribution on $\mathbb{R}^2$. \cite{bergsma2014} introduce $t^*$ not as a U-statistic but as the closely related biased V-statistic; we consider the U-statistic as it simplifies some of the computations in Sections \ref{sec:untied} and \ref{sec:tied} but present modifications to our algorithm that allow one to compute the V-statistic in Appendix \ref{sec:v_stat}. A comprehensive overview of U/V-statistics and their properties can be found in \citet{serfling1980}.

As noted by \cite{bergsma2014}, we may rewrite the function $a$ as
\begin{align}
\label{eq:a_as_ind}
\begin{split}
	a(z_1,z_2,z_3,z_4) &= I(z_1,z_3 < z_2,z_4) + I(z_1,z_3 > z_2,z_4) \\
	&\ \ - I(z_1,z_2<z_3,z_4) - I(z_1,z_2 > z_3,z_4)
\end{split}
\end{align}
where $I(z_1,z_2 < z_3,z_4)$ is the indicator of the event $\max(z_1,z_2) < \min(z_3,z_4)$. After rewriting $a$ in this way we see that computation of the $t^*$ statistic requires only knowledge of the relative positioning of the observations for which we make the following definitions.  Let $(x_1,y_1),...,(x_4,y_4)$ be four points relabelled so that $x_1\leq x_2\leq x_3\leq x_4$.  We then say that the points are 
\begin{align*}
	\left.
		\begin{array}{lll}
			\text{\emph{inseparable}} & \mbox{if } \ \ 
				\begin{array}{@{}l@{}}
					\text{$x_2=x_3$ or there exists a permutation $\pi$ of $\{1,2,3,4\}$} \\
					\text{so that $y_{\pi(1)}\leq y_{\pi(2)}=y_{\pi(3)}\leq y_{\pi(4)}$,}
				\end{array}
				\end{array}
					\right.
\end{align*}
and if they are not inseparable, then we call them
\begin{align*}
	\left\{
		\begin{array}{lll}
			\text{\emph{concordant}} & \mbox{if } \text{$\max(y_1,y_2) < \min(y_3,y_4)$\ \ or \ $\max(y_3,y_4) < \min(y_1,y_2)$,} \\
			\text{\emph{discordant}} & \mbox{if }  \text{$\max(y_1,y_2) > \min(y_3,y_4)$ and $\max(y_3,y_4) > \min(y_1,y_2) $.}
		\end{array}
	\right.
\end{align*}
These definitions categorize all quadruples of points, that is, any quadruple of points must be exactly one of inseparable, concordant, or discordant. Moreover, when all coordinates are distinct any collection of four points will be either concordant or discordant, see Figure \ref{fig:concordant}. 
%
We motivate calling points inseparable by noting that, in the
$x_2=x_3$ case, we cannot draw a line parallel to the $y$-axis that
separates the $x$ values into two groups. Similarly in the case of $y_{\pi(1)}\leq y_{\pi(2)}=y_{\pi(3)}\leq y_{\pi(4)}$, there exists no such line parallel to the $x$-axis that separates the $y$ values into two groups.

\begin{figure}
        \centering
         \subfloat[Concordant\label{fig:pos_con}]{%
		\includegraphics[width=.3\textwidth]{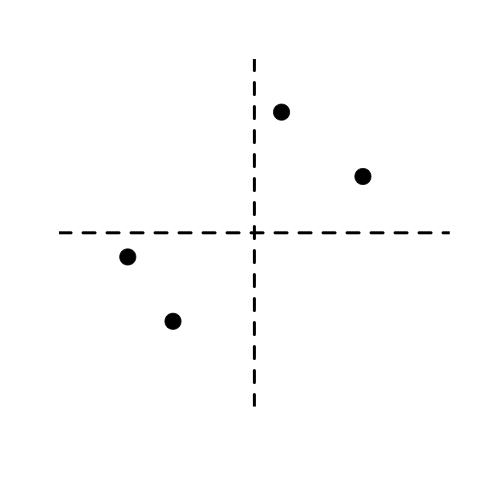} \includegraphics[width=.3\textwidth]{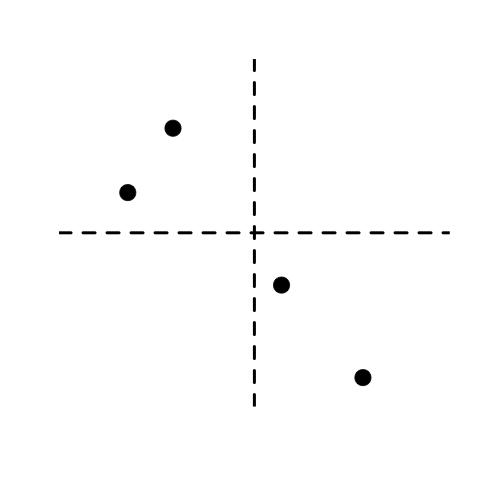}
	}
	\subfloat[Discordant\label{fig:dis_con}]{%
		\includegraphics[width=.3\textwidth]{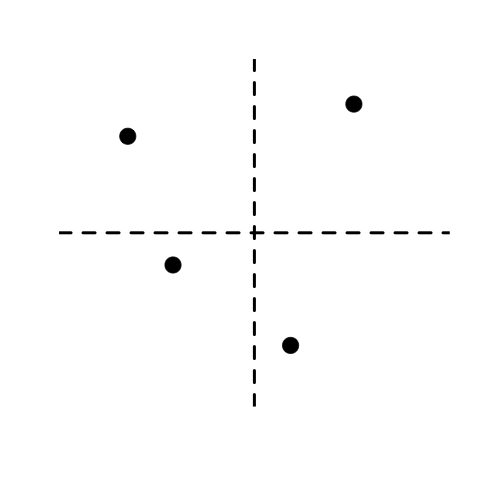}
	} 
    
        \caption{Relative position of points within quadrants does not matter, only that they remain in their respective quadrants.}
        \label{fig:concordant}
\end{figure}

We will derive two algorithms for the computation of $t^*$, the first works only in the case that the data contains no ties, that is all $x_1,...,x_n$ are distinct and similarly for $y_1,...,y_n$, and the second works for all data. While the second algorithm is strictly more general than the first it is also substantially complicated by the need to consider the case of inseparable points. We present the algorithm for data without ties in Section \ref{sec:untied} and give the general algorithm in Section \ref{sec:tied}.

\subsection{A Preliminary Lemma}

Before moving on, it will be useful to rewrite $t^*$ to capture a certain permutation invariance and state a basic, but very useful, lemma. Let $C(n,4)=\{\{i,j,k,l\}:1\leq i<j<k<l\leq n\}$, and $S_4$ be the set of permutations on 4 elements. For ease of notation, for any $\pi\in S_4$ and $(z_1,z_2,z_3,z_4)\in\R^4$ we  define $z_{\pi(1,2,3,4)} := (z_{\pi(1)},...,z_{\pi(4)})$.
We may then rewrite \eqref{eq:t_star} as
\begin{allowdisplaybreaks}
\begin{align*}
	t^* &= \frac{(n-4)!}{n!}\sum_{\{i,j,k,l\}\in C(n,4)}\sum_{\pi\in S_4}a(x_{\pi(i,j,k,l)})a(y_{\pi(i,j,k,l)}) \\
	&= \frac{(n-4)!}{n!}\sum_{\{i,j,k,l\}\in C(n,4)}b_{ijkl}, \numberthis \label{eq:t_star_rewrite}
\end{align*}
\end{allowdisplaybreaks}
where 
\[
	b_{ijkl} := \sum_{\pi\in S_4(i,j,k,l)} a(x_{\pi(i,j,k,l)})a(y_{\pi(i,j,k,l)})
\]
is clearly invariant to any permutation of $i,j,k,l$.

We now characterize the possible values $b_{ijkl}$ may take.

\begin{lem} \label{lem:b_values}
Let $A = \{(x_1,y_1),(x_2,y_2),(x_3,y_3),(x_4,y_4)\}\subset\mathbb{R}^2$. Then
\begin{align*}
	b_{1234} = \left\{
		\begin{array}{lll}
			16 & \mbox{if } \text{the points in $A$ are concordant}\\
			-8 & \mbox{if }  \text{the points in $A$ are discordant} \\ 
			0 & \mbox{if }  \text{the points in $A$ are inseparable}
		\end{array}
	\right.
\end{align*}
\end{lem}

The proof of Lemma~\ref{lem:b_values} is a straightforward but lengthy case-by-case analysis and we defer it to Appendix \ref{sec:b_values}.

\section{The Algorithm for Data Without Ties} \label{sec:untied}

Throughout this section we  assume that $(x_1,y_1),...,(x_n,y_n)$ contain no ties, that is, $x_1,...,x_n$ are pairwise distinct and so are $y_1,...,y_n$. As there are no ties, every quadruple of points is either concordant or discordant.  It follows from Equation \eqref{eq:t_star_rewrite} and Lemma \ref{lem:b_values} that
\begin{align*}
	t^* &= \frac{(n-4)!}{n!}\sum_{\{i,j,k,l\}\in C(n,4)}b_{ijkl} \\
	&= \frac{(n-4)!}{n!}\sum_{\{i,j,k,l\}\in C(n,4)} \Big[16\cdot I(\{\text{$(x_i,y_i),...,(x_l,y_l)$ are concordant}\}) \\
	&\ \ \ \ \ \ \ \ \ \ \ \ \ \ \ \ \ \ \ \ \ \ \ \ \ \ \ \ \ \ \ \  \ \ -8\cdot I(\{\text{$(x_i,y_i),...,(x_l,y_l)$ are discordant}\})\Big] \\
	&= \frac{(n-4)!}{n!}(16\cdot N_c - 8\cdot N_d) \\
	&= \frac{(n-4)!}{n!}(24\cdot N_c) - \frac{1}{3}, \numberthis \label{eq:simple_untied_t_star}
\end{align*}
where $N_c$ and $N_d$ are the numbers of concordant and discordant quadruples in $(x_1,y_1),...,(x_n,y_n)$, respectively, and the last equality holds since every quadruple of points is either concordant or discordant implying that ${n \choose 4} = N_d + N_c$. Thus computing $t^*$ requires only computing the number of concordant quadruples of points. We now show that this can be done efficiently.

Suppose we have relabeled the points so that $x_1<x_2<...<x_n$. Rewriting sums we have that
\begin{allowdisplaybreaks}
\begin{align*}
	&\ N_c = \sum_{1\leq i<j<k<l\leq n} I(\text{$(x_i,y_i),(x_j,y_j),(x_k,y_k),(x_l,y_l)$ are concordant})) \\
	&= \sum_{3\leq k\leq n-1} \sum_{k<l\leq n} \sum_{1\leq i < j<k} \hspace{-2mm} I(\text{$(x_i,y_i),(x_j,y_j),(x_k,y_k),(x_l,y_l)$ are concordant})) \\
	&= \sum_{3\leq k\leq n-1} \sum_{k<l\leq n} \sum_{1\leq i < j<k} \hspace{-2mm} I(\text{$y_i,y_j < y_k,y_l$}) + I(\text{$y_k,y_l < y_i,y_j$}) \\
	&= \sum_{3\leq k\leq n-1} \sum_{k<l\leq n}{M_<(k,l) \choose 2} + {M_>(k,l) \choose 2}
\end{align*}
\end{allowdisplaybreaks}
where we define
\begin{align*}
	M_<(k,l) &:= |\{i:\ 1\leq i<k,\ y_i < \min(y_k, y_l)\}|, \\
	M_>(k,l) &:= |\{i:\ 1\leq i<k,\ y_i > \max(y_k, y_l)\}|.
\end{align*}
The last line in the above summation is, effectively, the algorithm. Note that the summation is over $O(n^2)$ terms and, consequently, if we can find $M_<(k,l)$ and $M_>(k,l)$ in $O(\log(n))$ time then we have found an algorithm for computing $N_c$ in $O(n^2\log(n))$ time. To find $M_<(k,l)$ and $M_>(k,l)$ in $O(\log(n))$ time we use a binary tree data structure with an appropriate balancing algorithm to ensure that inserts and searching can be done in $O(\log(n))$ time.  One example of this type of data structure are the so-called red-black trees \citep{red_black}. In particular, given that we have inserted the values $y_1,y_2,...,y_{k-1}$ into a red-black tree we may insert another $y_k$ into the tree in $O(\log(k))$ time and a simple extension of the traditional red-black framework allows one, for any $y$, to find $|\{1\leq i\leq k-1: y_i<y\}|$ and $|\{1\leq i\leq k-1: y_i>y\}|$ in $O(\log(k))$ time.

Combining the above observations, Algorithm \ref{alg:untied} gives an $O(n^2\log(n))$ procedure for finding the number of concordant quadruples which is easily extended to a computation of $t^*$ via Equation \eqref{eq:simple_untied_t_star}. Note that in Algorithm \ref{alg:untied} there is a preprocessing step in which  we  sort the $x_1,...,x_n$ values in ascending order and then reorder the $y_i$ to match this new order. Since this preprocessing can be done in worst case $O(n\log(n))$ time with a number of algorithms, merge-sort for example, it is not a significant component of the overall asymptotic run time analysis.

 \begin{algorithm}[t]
\caption{}\label{alg:untied}
\begin{algorithmic}[1]
\Procedure{NumConcordant}{($x_1,y_1$),...,($x_n,y_n$)} 
\State $x\gets (x_1,....,x_n)$
\State $y\gets (y_1,....,y_n)$
\State Sort $x$ in ascending order and relabel $y$ to match this new order
\State $rbTree \gets $ empty red-black tree
\State $totalConcordant \gets 0$

\For{$k = 1,....,n-1$}
\For{$\ell = k+1,....,n$}
	\State $minY \gets min(y_k,y_\ell)$
	\State $maxY \gets max(y_k,y_\ell)$
	\State $numLess \gets $ number of elements $<minY$ in $rbTree$
	\State $numGreater \gets $ number of elements $>maxY$ in $rbTree$
	\State $totalConcordant = totalConcordant + {numLess \choose 2} + {numGreater \choose 2}$
\EndFor
\State Insert $y_k$ into $rbTree$
\EndFor
\State \Return $totalConcordant$
\EndProcedure
\end{algorithmic}
\end{algorithm}

\section{The General Algorithm} \label{sec:tied}

Now suppose that there are no restrictions on the values of $(x_1,y_1),...,(x_n,y_n)$ and that we have reordered the points so that $x_1\leq ... \leq x_n$. 
For any $3\leq k\leq l\leq n$, let
\begin{allowdisplaybreaks}
\begin{align}
	top(k,l) &= |\{1\leq i< k: x_i\not= x_k\ \text{and}\ y_i> \max(y_k,y_l)\} | ,  \label{eq:top} \\
	mid(k,l) &= |\{1\leq i< k: x_i\not= x_k\ \text{and} \min(y_k,y_l) < y_i< \max(y_k,y_l)\} |,\\
	bot(k,l) &= |\{1\leq i< k: x_i\not= x_k\ \text{and}\ y_i< \min(y_k,y_l)\} | ,\\
	\mathit{eqMin}(k,l) &= |\{1\leq i< k: x_i\not= x_k\ \text{and}\ y_i = \min(y_k,y_l)\}  |,\\
	\mathit{eqMax}(k,l) &= |\{1\leq i< k: x_i\not= x_k\ \text{and}\ y_i = \max(y_k,y_l)\} |. \label{eq:eqmax}
\end{align}
\end{allowdisplaybreaks}
These quantities correspond to a partitioning of the points $(x_i,y_i)$ with $i<k$ and $x_i\not=x_k$.  We illustrate this partitioning  in Figure \ref{fig:partition}.
\begin{figure}[t]
        \centering
        \includegraphics[width=.4\textwidth]{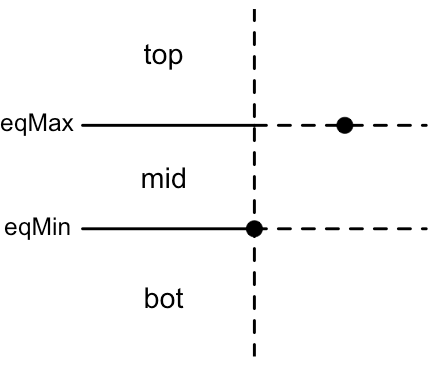} \caption{Partitioning of the points with $x$ value strictly less than two given points. Solid lines correspond to $\mathit{eqMax}$ and $\mathit{eqMin}$, the points whose $y$-values equal the maximum or minimum of the $y$ values of the two given points, respectively.}
        \label{fig:partition}
\end{figure}
For fixed $3\leq k<l\leq n$ we have, by Lemma \ref{lem:b_values} and since $x_1\leq ... \leq x_n$,
\begin{allowdisplaybreaks}
\begin{align*}
	\sum_{1\leq i<j< k }\hspace{-2mm} b_{ijkl} &= 16\cdot \underbrace{|\{1\leq i<j<k:\ \text{$i,j,k,l$ correspond to concordant points}\}|}_{:=N_{\rm con}(k,l)} \\
	&\ \ -8\cdot \underbrace{|\{1\leq i<j<k:\ \text{$i,j,k,l$ correspond to discordant points}\}|}_{:=N_{\rm dis}(k,l)} \\
	&= 16\cdot N_{\rm con}(k,l) -8\cdot N_{\rm dis}(k,l).
\end{align*}
\end{allowdisplaybreaks}
Hence, similarly as in the case without ties, we may write
\begin{allowdisplaybreaks}
\begin{align*}
	\sum_{\{i,j,k,l\}\in C(n,4)}b_{ijkl} &= \sum_{3\leq k\leq n-1} \sum_{k<l\leq n} \sum_{1\leq i < j<k} b_{ijkl} \\
	&= \sum_{3\leq k\leq n-1} \sum_{k<l\leq n} 16\cdot N_{\rm con}(k,l) -8\cdot N_{\rm dis}(k,l).
\end{align*}
\end{allowdisplaybreaks}
Again the last line of the above summation is effectively the algorithm.  
Since the sums are over $O(n^2)$ terms, if we can show that $N_{\rm con}(k,l)$ and $N_{\rm dis}(k,l)$ can be computed in $O(\log(n))$ time then we have obtained an $O(n^2\log(n))$ algorithm for computing $t^*$.   We show next that this is indeed possible, beginning with the observation that
\begin{align} \label{eq:num_con}
	N_{\rm con}(k,l) &= {top(k,l) \choose 2} + {bot(k,l) \choose 2},
\end{align}
if $y_k=y_l$ then 
\begin{align} \label{eq:num_dis_eq}
	N_{\rm dis}(k,l) = 0,
\end{align}
and if $y_k\not =y_l$ then 
\begin{align*}
	N_{\rm dis}(k,l) &= top(k,l)\cdot \left(mid(k,l) + \mathit{eqMin}(k,l)+ bot(k,l)\right) \numberthis \label{eq:num_dis_ineq}\\
	&+ bot(k,l)\cdot \left(mid(k,l) + \mathit{eqMax}(k,l)\right) \\
	&+ \mathit{eqMin}(k,l)\cdot (mid(k,l) + \mathit{eqMax}(k,l)) \\
	&+ \mathit{eqMax}(k,l)\cdot mid(k,l) \\
	& + {mid(k,l) \choose 2} - \sum_{y\in unique(k,l)}{|\{1\leq i< k: x_k\not= x_i\ \text{and}\ y_i=y\}| \choose 2} 
\end{align*}
where 
\begin{align*}
	unique(k,l) := \{y_i :1 \leq\kern-.2em i\kern-.2em <\kern-.2em k\ \text{and}\ x_i\not=x_k\ \text{and}\ \min(y_k,y_l)\kern-.2em <\kern-.2em y_i\kern-.2em <\kern-.2em \max(y_k,y_l)\}.
\end{align*}

Suppose we have a red-black tree into which we have inserted all $y_i$ with $1\leq i < k$ and $x_i\not=x_k$.  Then it is clear that the quantities in Equations \eqref{eq:top}-\eqref{eq:eqmax} can each be computed in $O(\log(k))$ time. Note that, unlike in the untied case, we require that the red-black tree not include any $y_i$ values corresponding to $x_i=x_k$; accomplishing this algorithmically is very simple: as we iterate across the $x_k$ values we delay inserting their associated $y_k$ values into the red-black tree until we reach a $x_{l}$ with $x_l\not=x_{l-1}$, upon reaching such an $x_l$ we insert all postponed $y$ values into the red-black tree and then restart the postponing of $y$ values starting with $y_l$.

We see that, as in the discussion of Algorithm \ref{alg:untied}, we can progressively compute almost all of the quantities in Equations \eqref{eq:num_con} and \eqref{eq:num_dis_ineq} with each iteration taking $O(\log(n))$ time. The only complication is the computation of
\begin{align} \label{eq:untied_comp_sum}
	\sum_{y\in unique(k,l)}{|\{1\leq i< k: x_k\not= x_i\ \text{and}\ y_i=y\}| \choose 2}
\end{align}
which corresponds to all quadruples of points $(x_i,y_i),(x_j,y_j),(x_k,y_k),(x_l,y_l)$ for which $\min(y_k,y_l)<y_i=y_j<\max(y_k,y_l)$.  These are inseparable and are being over-counted by ${mid(k,l) \choose 2}$. Note that this summation is in the reverse order of what we would like in order to simply generalize Algorithm \ref{alg:untied}. In particular, there is a condition on the $y$ values corresponding to $i$ and $j$ which is suppressed by the aggregate counts available from a query on a red-black tree.
We have, however, already established a methodology to count values such as \eqref{eq:untied_comp_sum}. In particular, note that
\begin{allowdisplaybreaks}
\begin{align}
	 &\sum_{3\leq k<l\leq n} \sum_{y\in unique(k,l)}{|\{1\leq i< k: x_k\not= x_i\ \text{and}\ y_i=y\}| \choose 2} \nonumber \\
	 & = \sum_{1\leq i < j<k<l\leq n} I(\{ x_k\not= x_j \ \text{and}\ y_i=y_j \ \text{and}\ \min(y_k,y_l) <y_i < \max(y_k,y_l) \}) \nonumber \\
	 &  = \sum_{1\leq i < j \leq n-2} I(\{y_i=y_j\})\cdot \underbrace{|\{k: j< k \leq n \ \text{and}\ x_k\not= x_j\ \text{and}\ y_j < y_k\}|}_{:=\ top^*(j)} \nonumber \\
	 &\hspace{40 mm} \cdot \underbrace{|\{k: j< k \leq n\ \text{and}\ x_k\not= x_j\ \text{and}\ y_j > y_k\}|}_{:=\ bot^*(j)} \nonumber \\
	 & = \displaystyle\sum_{j\in\{n-2,n-1,...,2\}}\sum_{i\in\{j-1,j-2,...,1\}} I(\{y_i=y_j\})\cdot top^*(j) \cdot bot^*(j). \label{eq:dis_extra_con}
\end{align}
\end{allowdisplaybreaks}
It follows that all that is needed to compute the total contribution of the term in Equation \eqref{eq:untied_comp_sum} to $t^*$ is to run a modified version of Algorithm \ref{alg:untied} across the data in reverse order.  Our algorithm becomes the following:
\begin{enumerate}[(i)]
	\item Perform a first pass across the data where we ignore the effect of \eqref{eq:untied_comp_sum} and count all other quantities.
	\item Perform a second pass across the data in reverse order to compute \eqref{eq:dis_extra_con}.
	\item Appropriately combine the results of (i) and (ii) to obtain $t^*$.
\end{enumerate}
This amounts to over-counting discordant quadruples on a first pass and then undoing this over-counting on a second pass. Since both of these passes over the data require $O(n^2\log(n))$ time,  our general Algorithm \ref{alg:tied}, which leverages the above observations, computes $t^*$ in $O(n^2\log(n))$ time. An implementation of Algorithm \ref{alg:tied} is available in the R package TauStar accessible via CRAN, the Comprehensive R Archive Network\footnote{See \url{https://cran.r-project.org/web/packages/TauStar/index.html}} \citep*{R, taustar-package}.

 \begin{algorithm}[t]
 \caption{~\\ Algorithm for efficiently computing $t^*$ in the general case. Comments are displayed in {\color{commentcol} gray}.}\label{alg:tied}
\begin{algorithmic}[1]
\Procedure{$t^*$}{($x_1,y_1$),...,($x_n,y_n$)} 
\State $x\gets (x_1,....,x_n)$
\State $y\gets (y_1,....,y_n)$
\State Sort $x$ in ascending order and relabel $y$ to match this new order
\State $rbTree \gets $ empty red-black tree \CComment{Used in first pass through data.}
\State $savedYValues \gets$ empty list \CComment{A list to store y values whose insertion into the red-black tree is delayed.}
\State $totalConcordant \gets 0$ \CComment{Total concordant quadruples counted so far.}
\State $totalDiscordant \gets 0$ \CComment{Total discordant quadruples counted so far.}
\For{$k \gets 1,....,n-1$} \label{line:loop1}
	\If{$k \not= 1$ and $x_{k-1} \not= x_k$}
		\CComment{If $k\not=1$ and $x_k\not= x_{k-1}$ insert all delayed $y$ values into the red-black tree. In any case, save $y_k$ to be inserted in the tree on some future iteration.} \label{line:delayed}
		\For{$yVal$ in $savedYValues$}
			\State Insert $yVal$ into $rbTree$
		\EndFor
		\State Empty the list $savedYValues$
	\EndIf
	\State Append $y_k$ to $savedYValues$ \label{line:beforeloop2}
	\For{$\ell \gets k+1,..., n$}
		\CComment{Loop over $\ell > k$ and use \eqref{eq:num_con}, \eqref{eq:num_dis_eq},\eqref{eq:num_dis_ineq} while ignoring contributions of \eqref{eq:untied_comp_sum}.} 
		\State $minY \gets min(y_k,y_\ell)$
		\State $maxY \gets max(y_k,y_\ell)$
		\State $top \gets $ number of elements $>maxY$ in $rbTree$
		\State $mid \gets $ number of elements $<maxY$ and $>minY$ in $rbTree$
		\State $bot \gets $ number of elements $<minY$ in $rbTree$
		\State $\mathit{eqMin} \gets$ number of elements equal to $minY$ in $rbTree$
\State $\mathit{eqMax} \gets$ number of elements equal to $maxY$ in $rbTree$
		\State $totalConcordant \gets totalConcordant + {top \choose 2} + {bot \choose 2}$ \label{line:total_con}
		\If{$minY \not= maxY$}
			\State $totalDiscordant \gets$ \parbox[t]{.5\linewidth}{$totalDiscordant + {mid \choose 2} + top\cdot mid + top\cdot bot + mid\cdot bot + \mathit{eqMin}\cdot (top + mid + \mathit{eqMax}) + \mathit{eqMax}\cdot(mid + bot)$}
		\EndIf
	\EndFor
\EndFor
\LineComment{\parbox[t]{.9\linewidth}{In the next loop we will run along the data in reverse to undo the over-counting resulting from ignoring the contribution of \eqref{eq:untied_comp_sum}.}}
\State Empty the list $savedYValues$
\State $revRbTree \gets $ empty RB tree \CComment{Used in second pass over the data.}
\For{$j \gets n,....,2$}
\algstore{tiedalg}
\end{algorithmic}
\end{algorithm}

\begin{algorithm}[]
\begin{algorithmic}[1]
\algrestore{tiedalg}
\If{$j \not= n$ and $x_{j+1} \not= x_j$}
		\CComment{Inserting the delayed values similarly as in Line \ref{line:delayed}.}
		\For{$yVal$ in savedYValues}
			\State Insert $yVal$ into $revRbTree$
		\EndFor
		\State Empty the list $savedYValues$
	\EndIf
	\State Append $y_j$ to $savedYValues$
	\For{$i \gets j-1,....,1$} \CComment{Use \eqref{eq:dis_extra_con} to compute the number of over counts.}
		\State $minY \gets min(y_i,y_j)$
		\State $maxY \gets max(y_i,y_j)$
		\State $top \gets $ number of elements $>maxY$ in $revRbTree$
		\State $bot \gets $ number of elements $<minY$ in $revRbTree$ 
		\If{$minY = maxY$}
			\State $totalDiscordant \gets totalDiscordant - top\cdot bot$
		\EndIf
	\EndFor
\EndFor
\State \Return $\frac{1}{n(n-1)(n-2)(n-3)}(16\cdot totalConcordant - 8\cdot totalDiscordant)$ \label{line:return}
\EndProcedure
\end{algorithmic}
\end{algorithm}

\section{Simulations}\label{sec:sim}
We test the run times of Algorithm \ref{alg:tied} and a na\"{\i}ve implementation, both written in C++ and available through R in previously mentioned TauStar package, for sample sizes $n$ ranging from 100 to 300; the implementation of Algorithm \ref{alg:tied} uses the red-black tree C library of \citet{martinian}. The results of these simulations are presented in Table \ref{tab:run_times}. As the table shows, the $O(n^4)$ running time of the na\"{\i}ve algorithm becomes already a practical concern for sample sizes in the hundreds while Algorithm \ref{alg:tied} is essentially instant for such sample sizes. Table \ref{tab:run_times_fast} provides a perspective on the run time of Algorithm \ref{alg:tied} for substantially larger samples.

\begin{table}[t]
\caption{Run times of the na\"{\i}ve algorithm and Algorithm \ref{alg:tied} for various sample sizes (in seconds and averaged over 10 samples).} \label{tab:run_times}
\begin{tabular}{|c|ccccc|} \hline
Sample Size & 100 & 150 & 200 & 250 & 300 \\ \hline
Algorithm \ref{alg:tied} & 0.0009 & 0.0023 & 0.0043 & 0.0072 & 0.01 \\
Na\"{\i}ve Algorithm & 0.287 &  1.55 &  5.58 & 14.34 & 28.95 \\ \hline
\end{tabular}
\end{table}

\begin{table}[t]
\caption{Run times of Algorithm \ref{alg:tied} for larger sample sizes (in seconds and  averaged over 10 samples).} \label{tab:run_times_fast}

\begin{tabular}{|c|ccccc|} \hline
Sample Size & 1000 &  3250 & 5500 & 7750 & 10000 \\ \hline
Algorithm \ref{alg:tied} & 0.1265 &  1.7354 &  5.2744 &  11.1833 & 19.115 \\ \hline
\end{tabular}
\end{table}

It is possible to approximate $t^*$, or in other words, estimate $\tau^*$ by a Monte-Carlo subsampling procedure where, for small $m < n$, subsets of size $m$ are repeatably selected from the data at random and the value of $t^*$ on each of these subsets is then averaged. Indeed, the case of $m=4$ is a strategy suggested by \citet{bergsma2014}. While our algorithm makes the computation of $t^*$ on moderate to large samples feasible, an approximate strategy will be necessary for very large samples. Unfortunately, resampling procedures require choosing a number of resampling iterations and, as is shown by Table \ref{tab:resample}, choosing too few iterations can result in a estimator with large variance.  Table \ref{tab:resample} also suggests that a choice of $m>4$ may be useful.\footnote{R code to reproduce the results of Tables \ref{tab:run_times}-\ref{tab:resample} can be found on the first author's webpage: \url{http://www.stat.washington.edu/~lucaw/public_resources/eff_comp_2015/tables.R}}

\newsavebox{\smallvec}
\savebox{\smallvec}{$\left(\begin{smallmatrix}0\\0\end{smallmatrix}\right)$}
\newsavebox{\smlmat}
\savebox{\smlmat}{$\left(\begin{smallmatrix}1&0\\0&1\end{smallmatrix}\right)$}
\begin{table}[t]
\caption{Sample variance of resampling-based estimates of $\tau^*$ relative to the sample variance of $t^*$ computed for all data. Here relative variance is the ratio of the former and the latter variance.  The variances are computed from 1000 samples of size $n=1000$.  Resampled subsets were of size $m \in \{4,30\}$.  The samples were drawn as pairs of independent $N(0,1)$ random variables.
} \label{tab:resample}

\begin{tabular}{| >{\centering\arraybackslash} m{6em} | >{\centering\arraybackslash} m{2.5em} >{\centering\arraybackslash} m{2em} >{\centering\arraybackslash} m{2em} >{\centering\arraybackslash} m{2em} >{\centering\arraybackslash} m{2em} >{\centering\arraybackslash} m{2em} >{\centering\arraybackslash} m{2em} |} \hline
\# Resamples & 200 & 400 & 800 & 1600 & 3200 & 6400 & 12800 \\ \hline
Relative Var. ($m=4$) & 3932.84 & 2118.1 & 911.67 & 472.67 & 230.82 & 115.57 & 57.03 \\
Relative Var. ($m=30$) & 8.24 & 4.19 & 2.43 & 1.67 & 1.21 & 1.11 & 1.04 \\ \hline
\end{tabular}
\end{table}


\section{Conclusion}
\label{sec:conclusion}
We have presented an algorithm which computes the $U$-statistic $t^*$ corresponding to the $\tau^*$ sign covariance of \citet{bergsma2014} in $O(n^2\log(n))$ time, substantially outperforming a na\"{\i}ve implementation.  The computational savings in our algorithm are driven by the use of binary trees and the permutation invariance inherent in $t^*$ (recall Lemma~\ref{lem:b_values}).

\appendix

\section{Modifications for the V-Statistic }\label{sec:v_stat}

This section provides an overview of necessary modifications to Algorithm \ref{alg:tied} in order to compute the V-statistic version of $t^*$. Suppose, as usual, that we have reordered the pairs $(x_1,y_1),...,(x_n,y_n)$ so that $x_1\leq x_2\leq ...\leq x_n$. Then the V-statistic for $\tau^*$ is
\begin{align*}
	&t_V^* = \frac{1}{n^4}\sum_{1\leq i,j,k,l\leq n} a(x_i,x_j,x_k,x_l)a(y_i,y_j,y_k,y_l) \\
	&= \frac{1}{n^4}\left(\sum_{1\leq i< j< k< l\leq n} b_{ijkl} + \sum_{1\leq i< j< k\leq n} \frac{b_{ijkk} + b_{ijjk} + b_{iijk}}{2} + \sum_{1\leq i < k\leq n}\frac{b_{iikk}}{4}\right)\\
	&= \frac{1}{n^4}\left(\sum_{1\leq i< j< k< l\leq n} b_{ijkl} +  \sum_{1\leq i< j< k\leq n} \frac{b_{ijkk}+ b_{iijk}}{2} + \sum_{1\leq i < k\leq n}\frac{b_{iikk}}{4}\right).
\end{align*}
Here, the second equality holds since $a(x_i,x_j,x_k,x_l)a(y_i,y_j,y_k,y_l)=0$ if any three of $i,j,k,l$ are equal.  The third equality holds because $b_{ijjk}=0$ for all $i< j< k$; indeed, $x_i\leq x_j \leq x_k$ implies that $b_{ijjk}$ corresponds to an inseparable collection of points. Note that, in the above equations, we have coefficients of $\frac{1}{2}$ on $b_{ijkk},b_{iijk}$ and $\frac{1}{4}$ on $b_{iikk}$, these are corrective factors to account for the fact that the number of permutations of four elements where exactly two are equal is $|S_4|/2$ while the number of permutations where exactly two pairs of two are equal is $|S_4|/4$. Now we may continue to rewrite $t^*_V$ as
\begin{align*}
	&t_V^* = \frac{1}{n^4} \left( \sum_{1\leq i< j< k< l\leq n} b_{ijkl} + \sum_{1\leq i< j< k\leq n} \frac{b_{ijkk}+ b_{iijk}}{2} + \sum_{1\leq i < k\leq n}\frac{b_{iikk}}{4} \right)\\ 
	&= \frac{1}{n^4} \left( \sum_{1\leq i< j< k< l\leq n} b_{ijkl} + \sum_{1\leq i< j< k\leq n} \frac{b_{ijkk}}{2}+ \sum_{1\leq i< k< l\leq n}\frac{b_{iikl}}{2} + \sum_{1\leq i < k\leq n}\frac{b_{iikk}}{4} \right)\\ 
	&= \frac{1}{n^4}\sum_{3\leq k\leq n} \Bigg(\sum_{k<l\leq n}\left(\sum_{1\leq i < j<k} b_{ijkl} + \sum_{1\leq i < k}\frac{b_{iikl}}{2}\right) + \sum_{1\leq  i<  j < k} \frac{b_{ijkk}}{2} + \sum_{1\leq  i < k} \frac{b_{iikk}}{4} \Bigg).
\end{align*}
If $k=n$ then $\sum_{k<l\leq n}$ is the empty sum which we define to equal 0. For a fixed $k<l$ we know already, from Section \ref{sec:tied}, how to compute $\sum_{1\leq i < j<k} b_{ijkl}$ efficiently using a red-black tree and since $b_{iikl},b_{ijkk}$, and $b_{iikk}$ can only correspond to inseparable or concordant quadruples it is easy to see that
\begin{align}
	\sum_{1\leq i <k} \frac{1}{2} b_{iikl} &= 8\cdot (top(k,l) + bot(k,l)) \label{eq:iikl},\\ 
	\sum_{1\leq i < j <k} \frac{1}{2} b_{ijkk} &= 8\cdot \left({top(k,k) \choose 2} + {bot(k,k)\choose 2}\right) \label{eq:ijkl},\\ 
	\sum_{1\leq i <k} \frac{1}{4}b_{iikk} &= 4\cdot (top(k,k) + bot(k,k)) \label{eq:iikk}.
\end{align}
Thus we may compute $t_V^*$ by running Algorithm \ref{alg:tied} with the following modifications:
\begin{enumerate}[(i)]
	\item Change line \ref{line:loop1} to
	\begin{algorithmic}[1]
		\For{$k = 1,....,n$}
		\EndFor
	\end{algorithmic}
	This corresponds to the outer sum of $t^*_V$.
	\item After line \ref{line:beforeloop2} add the lines:
	\begin{algorithmic}[1]
		\State $top \gets$ number of elements $>y_k$ in $rbTree$
		\State $bot \gets$ number of elements $<y_k$ in $rbTree$
		\State $totalConcordant \gets totalConcordant + \frac{1}{2}\left({top\choose 2} + {bot\choose 2}\right) + \frac{1}{4}(top + bot)$
	\end{algorithmic}
	This accounts for the effect of \eqref{eq:ijkl} and \eqref{eq:iikk}.
	
	\item Change line \ref{line:total_con} to
	\begin{algorithmic}[1]
		\State $totalConcordant \gets totalConcordant + {top \choose 2} + {bot \choose 2} + \frac{1}{2}(top + bot)$
	\end{algorithmic}
	This corresponds to \eqref{eq:iikl}.
	
	\item Change line \ref{line:return} to
	\begin{algorithmic}[1]
	\State \Return $\frac{1}{n^4}(16\cdot totalConcordant - 8\cdot totalDiscordant)$
	\end{algorithmic}
\end{enumerate}
Finally, note that this Algorithm for computing $t^*_V$ clearly remains $O(n^2\log(n))$.
\section{Proof of Lemma \ref{lem:b_values}}\label{sec:b_values}

By permutation invariance, suppose we have relabeled so that $x_1\leq x_2\leq x_3\leq x_4$. We have 3 cases: \\
\begin{enumerate}[(i)]
	\item {\em The points in $A$ are inseparable.}	
	The fact that $b_{1234}=0$ is an immediate consequence of Equation \eqref{eq:a_as_ind}. \\

	\item {\em The points in $A$ are concordant.}	
	In this case we must have that $x_2<x_3$ and either $\max(y_1,y_2) < \min(y_3,y_4)$ or $\min(y_1,y_2) > \max(y_3,y_4)$. By symmetry we need only consider the case when $\max(y_1,y_2) < \min(y_3,y_4)$. By Equation \eqref{eq:a_as_ind} it follows, with some thought, that $a(x_{\pi(1,2,3,4)}) = a(y_{\pi(1,2,3,4)})$ for all permutations $\pi\in S_4$ and thus, for any $\pi\in S_4$ we have $a(x_{\pi(1,2,3,4)})a(y_{\pi(1,2,3,4)}) =  a(x_{\pi(1,2,3,4)})^2$ with
	\begin{align*}
 a(x_{\pi(1,2,3,4)})^2 
		&=\left\{
		\begin{array}{lll}
			1 & \mbox{if }  \text{$\max(x_{\pi(1)}, x_{\pi(2)}) < \min(x_{\pi(3)}, x_{\pi(4)})$ or} \\
			&\phantom{\mbox{if }} \text{$\min(x_{\pi(1)}, x_{\pi(2)}) > \max(x_{\pi(3)}, x_{\pi(4)})$ or}
				\\
				&\phantom{\mbox{if }} \text{$\max(x_{\pi(1)}, x_{\pi(3)}) < \min(x_{\pi(2)}, x_{\pi(4)})$ or}
				\\
				&\phantom{\mbox{if }} \text{$\min(x_{\pi(1)}, x_{\pi(3)}) > \max(x_{\pi(2)}, x_{\pi(4)})$},
				\\
			0 & \text{otherwise.}
		\end{array}
	\right.
	\end{align*}
	But since $x_1\leq x_2 < x_3\leq x_4$ we have that $a(x_{\pi(1,2,3,4)})a(y_{\pi(1,2,3,4)})=1$ if and only if $\{\pi(1),\pi(2)\} \in \{\{1,2\}, \{3,4\}\}$ or $\{\pi(1),\pi(3)\} \in \{\{1,2\}, \{3,4\}\}$. There are exactly $2^4=16$ such permutations and thus $b_{1234}=16$. \\
	
	\item {\em The points in $A$ are discordant.}	
	Once again we must have that $x_2<x_3$. It then follows, by the definition of discordant, that $y_1\not=y_2$ and $y_3\not=y_4$. We prove an intermediary lemma:
	
	\begin{lem}\label{lem:flip}
		Suppose that $(x_1,y_1),...,(x_4,y_4)$ are discordant and $x_1\leq x_2< x_3\leq x_4$. Let
		\begin{align*}
			(x_5,y_5) &= (x_1,y_2), &
			(x_6,y_6) &= (x_2,y_1), &
			(x_7,y_7) &= (x_3,y_3), &
			(x_8,y_8) &= (x_4,y_4),
		\end{align*}
		so that $(x_5,y_5),...,(x_8,y_8)$ are simply
                $(x_1,y_1),...,(x_4,y_4)$ with $y_1,y_2$
                switched. Then $b_{1234} = b_{5678}$. Moreover, the
                same result is true if we flipped $y_3,y_4$ instead of
                $y_1,y_2$.
	\end{lem}
	
	\begin{proof}
		First note that, trivially,  $a(x_{\pi(1,2,3,4)}) = a(x_{\pi(5,6,7,8)})$  for any $\pi\in S_4$.   Let $\pi$ be any permutation so that $a(x_{\pi(1,2,3,4)})^2=1$.  From case (ii) we know that we must have $\{\pi(1),\pi(2)\} \in \{\{1,2\}, \{3,4\}\}$ or $\{\pi(1),\pi(3)\} \in \{\{1,2\}, \{3,4\}\}$. Suppose that $\{\pi(1),\pi(2)\}=\{1,2\}$, and let $\pi'\in S_4$ be the permutation where
		\begin{align*}
			\pi'(1) = \pi(2),\ \ \pi'(2) = \pi(1),\ \ \pi'(3) = \pi(3),\ \ \pi'(4) = \pi(4).
		\end{align*}
		Then clearly 
		$
			a(x_{\pi(1,2,3,4)}) = a(x_{\pi'(1,2,3,4)}) = a(x_{\pi(5,6,7,8)}) = a(x_{\pi'(5,6,7,8)})$
		but
		\begin{align*}
			a(y_{\pi(1,2,3,4)}) &= a(y_{\pi'(5,6,7,8)}), &
			a(y_{\pi'(1,2,3,4)}) &= a(y_{\pi(5,6,7,8)}),
		\end{align*}
		and thus
		\begin{multline*}
			a(x_{\pi(1,2,3,4)})a(x_{\pi(1,2,3,4)}) + a(x_{\pi'(1,2,3,4)})a(x_{\pi'(1,2,3,4)}) \\
			= a(x_{\pi(5,6,7,8)})a(x_{\pi(5,6,7,8)}) + a(x_{\pi'(5,6,7,8)})a(x_{\pi'(5,6,7,8)}).
		\end{multline*}
		As we may perform a similar procedure to all $\pi\in S_4$ with $a(x_{\pi(1,2,3,4)})^2=1$ (changing the choice of $\pi'$),  we see that $b_{1234} = b_{5678}$ as claimed.
		
		Finally, pairing $\pi$ with $\pi'$ given by
		\begin{align*}
			\pi'(1) = \pi(1),\ \ \pi'(2) = \pi(2),\ \ \pi'(3) = \pi(4),\ \ \pi'(4) = \pi(3)
		\end{align*}
		shows that this result still holds if we had flipped $y_3,y_4$ instead of $y_1,y_2$.
	\end{proof}
	By Lemma \ref{lem:flip}, we may assume that $x_1\leq x_2<x_3\leq x_4$ and $y_1<y_2$ and $y_3<y_4$. Note that, by the definition of discordant, we must have that $y_2 > y_3$ and $y_1<y_4$. From case (ii) we know that there are only 16 permutations $\pi$ for which $a(x_{\pi(1,2,3,4)}) \not= 0$ and they satisfy
	\begin{align*}
	\{\pi(1),\pi(2)\} \in \{\{1,2\}, \{3,4\}\}\text{ or }\{\pi(1),\pi(3)\} \in \{\{1,2\}, \{3,4\}\}.
	\end{align*}
	 If  $\{\pi(1),\pi(2)\} \in \{\{1,2\}, \{3,4\}\}$ and
         $\{\pi(1),\pi(3)\} \in \{\{1,4\},\{2,3\}\}$, then  we have $a(y_{\pi(1,2,3,4)}) = 0$.  Similarly, $a(y_{\pi(1,2,3,4)}) = 0$ if $\{\pi(1),\pi(3)\} \in \{\{1,2\}, \{3,4\}\}$ and $\{\pi(1),\pi(2)\} \in \{\{1,4\},\{2,3\}\}$. This leaves only 8 permutations $\pi\in S_4$ for which $a(x_{\pi(1,2,3,4)})a(y_{\pi(1,2,3,4)})$ may be non-zero, and we check these explicitly:
	 \begin{align*}
		a(x_{1,2,3,4})a(y_{1,2,3,4}) &= -1\cdot 1 = -1, & a(x_{2,1,4,3})a(y_{2,1,4,3}) &= -1\cdot 1 = -1, \\
		a(x_{3,4,1,2})a(y_{3,4,1,2}) &= -1\cdot 1 = -1, & a(x_{4,3,2,1})a(y_{4,3,2,1}) &= -1\cdot 1 = -1, \\
		a(x_{1,3,2,4})a(y_{1,3,2,4}) &= 1\cdot -1 = -1, & a(x_{2,4,1,3})a(y_{2,4,1,3}) &= 1\cdot -1 = -1, \\
		a(x_{3,1,4,2})a(y_{3,1,4,2}) &= 1\cdot -1 = -1, & a(x_{4,2,3,1})a(y_{4,2,3,1}) &= 1\cdot -1 = -1.
	\end{align*}
	We conclude that $b_{1234} = -8$ as claimed.	 
\end{enumerate}


\bibliographystyle{abbrvnat}
\bibliography{kendall_tau}

\begin{thebibliography}{9}
\providecommand{\natexlab}[1]{#1}
\providecommand{\url}[1]{\texttt{#1}}
\expandafter\ifx\csname urlstyle\endcsname\relax
  \providecommand{\doi}[1]{doi: #1}\else
  \providecommand{\doi}{doi: \begingroup \urlstyle{rm}\Url}\fi

\bibitem[Bergsma and Dassios(2014)]{bergsma2014}
W.~Bergsma and A.~Dassios.
\newblock A consistent test of independence based on a sign covariance related
  to {K}endall's tau.
\newblock \emph{Bernoulli}, 20\penalty0 (2):\penalty0 1006--1028, 2014.

\bibitem[Christensen(2005)]{fast_kendall}
D.~Christensen.
\newblock Fast algorithms for the calculation of {K}endall's {$\tau$}.
\newblock \emph{Comput. Statist.}, 20\penalty0 (1):\penalty0 51--62, 2005.

\bibitem[Guibas and Sedgewick(1978)]{red_black}
L.~J. Guibas and R.~Sedgewick.
\newblock A dichromatic framework for balanced trees.
\newblock In \emph{19th Annual Symposium on Foundations of Computer Science},
  pages 8--21, Oct 1978.

\bibitem[Kendall(1938)]{kendall}
M.~G. Kendall.
\newblock A new measure of rank correlation.
\newblock \emph{Biometrika}, 30\penalty0 (1/2):\penalty0 pp. 81--93, 1938.

\bibitem[Martinian(2005)]{martinian}
E.~Martinian.
\newblock Red-black tree {C} code.
\newblock \url{http://web.mit.edu/\textasciitilde
  emin/www.old/source_code/red_black_tree/index. html}, 2005.

\bibitem[{R Core Team}(2015)]{R}
{R Core Team}.
\newblock \emph{R: A language and environment for statistical computing}.
\newblock R Foundation for Statistical Computing, Vienna, Austria, 2015.
\newblock URL \url{https://www.R-project.org/}.

\bibitem[Serfling(1980)]{serfling1980}
R.~J. Serfling.
\newblock \emph{Approximation theorems of mathematical statistics}.
\newblock John Wiley \& Sons, Inc., New York, 1980.
\newblock ISBN 0-471-02403-1.
\newblock {W}iley Series in Probability and Mathematical Statistics.

\bibitem[Spearman(1904)]{spearman}
C.~Spearman.
\newblock The proof and measurement of association between two things.
\newblock \emph{The American Journal of Psychology}, 15:\penalty0 72--101,
  1904.

\bibitem[Weihs(2015)]{taustar-package}
L.~Weihs.
\newblock \emph{TauStar: Efficient computation of the t* statistic of Bergsma
  and Dassios (2014)}, 2015.
\newblock URL \url{http://CRAN.R-project.org/package=TauStar}.
\newblock {R} package version 1.0.0.

\end{thebibliography}
\end{document}